\documentclass[twocolumn]{article}
\usepackage{amsfonts, amsmath, amssymb, enumerate, amsthm, mathtools}

\usepackage{graphicx, authblk}
\usepackage{color}
\usepackage{dsfont}
\definecolor{nb}{rgb}{.6,.176,1}
\definecolor{sienna}{rgb}{1,0,0}
\definecolor{darkgreen}{rgb}{0,.5,0}

\usepackage{graphicx}
\usepackage{tikz, enumerate, multicol}
\usepackage{float, cite}
\usetikzlibrary{automata, positioning}
\usetikzlibrary{shapes.geometric, arrows, shapes.misc}

\usepackage{amsfonts}
\usepackage{latexsym}
\usepackage{epsfig}    
\usepackage{amsmath}
\usepackage{amsthm}
\usepackage{graphicx}
\usepackage{subfig}
\usepackage{grffile}

\usepackage[utf8]{inputenc}
\usepackage[english]{babel}
 
\usepackage[left=0.8in,right=0.8in,top=1.1in,bottom=0.8in]{geometry}

\newtheorem{theorem}{Theorem}
\newtheorem{proposition}{Proposition}
\newtheorem{lemma}{Lemma}
\newtheorem{remark}{Remark}
\newtheorem{defi}{Definition}
\newtheorem{assumption}{Assumption}

 \newcommand {\st}{\mathcal{S}}

 \newcommand{\tp}{\tilde{p}}
 \newcommand{\tq}{\tilde{q}}
  \newcommand{\tr}{\tilde{r}}

\begin{document}

 \renewcommand{\labelenumi}{{(\roman{enumi})}}

 \title{{\bf Influencing Opinions of Heterogeneous Populations over Finite Time Horizons}}

\author[$\star$]{Arunabh Saxena}
\author[$\star$]{Bhumesh Kumar}
\author[$\star$]{Anmol Gupta}
\author[$\dagger$]{Neeraja Sahasrabudhe}
\author[$\star$]{Sharayu Moharir}

\affil[$\star$]{Department of Electrical Engineering, Indian Institute of Technology Bombay}
\affil[$\dagger$]{Department of Mathematical Sciences, Indian Institute of Science Education and Research Mohali}

\date{}

  \maketitle

\begin{abstract} 
	
	In this work, we focus on strategies to influence the opinion dynamics of a well-connected society. We propose a generalization of the popular voter model. This variant of the voter model can capture a wide range of individuals including strong-willed individuals whose opinion evolution is independent of their neighbors as well as conformist/rebel individuals who tend to adopt the opinion of the majority/minority. 
	
	Motivated by political campaigns which aim to influence opinion dynamics by the end of a fixed deadline, we focus on influencing strategies for finite time horizons. We characterize the nature of optimal influencing strategies as a function of the nature of individuals forming the society. Using this, we show that for a society consisting of predominantly strong-willed/rebel individuals, the optimal strategy is to influence towards the end of the finite time horizon, whereas, for a society predominantly consisting of conformist individuals who try to adopt the opinion of the majority, it could be optimal to influence in the initial phase of the finite time horizon.
	
%	We consider a discrete-time binary opinion dynamics on a complete graph, evolving over a fixed time interval $[0, T]$. This is a variant of the voter model and captures a more general behaviour of individuals, including strong-willed individuals whose opinion evolution is independent of their neighbours as well as conformist individuals who tend to adopt the opinion of the majority. The main aim of the paper is to come up with optimal influencing strategies for an external influencing agency with a fixed time budget. An optimal influencing strategy is the one that maximizes the expected number of people with favourable opinion at the end of time $T$. We also discuss some variation of these models and explore corresponding optimal strategies. 
\end{abstract}
%\keywords{Random walk on free group, random walk in random  environment, trees, transience, Central Limit Theorem, Positive Speed}

\section{Introduction}
\label{Sec:Intro}

Opinion dynamics have been a subject of study in various fields including sociology, philosophy, mathematics, and physics for a very long time \cite{xia2011opinion}. In this work, we focus on a variant of a widely studied binary opinion dynamics model known as the voter model \cite{holley1975ergodic, clifford1973model}. In the voter model, society is modeled using a graph where each individual is a node and edges represent links between these individuals. Each individual holds one of two possible opinions, e.g., pro-government and anti-government. The opinions of individuals evolve over time. Assuming time is slotted, one individual is chosen uniformly at random at the beginning of each time-slot. This individual then adopts the opinion of one of its neighbors, chosen uniformly at random. The voter model is a useful framework to study opinion dynamics and the spread of competing epidemics. Variants and generalizations of the voter models have also been studied \cite{yildiz2011discrete, majmudar2015voter}. {\let\thefootnote\relax\footnote{A preliminary version of this work appeared in  \cite{MTNS}.}}

Our model generalizes the classical voter model. In each time-slot, the opinion of the selected individual evolves according to a general linear function of the opinion of its neighbors. %This differs from the original voter model, where the selected individual adopts the opinion of one of its neighbors chosen uniformly at random. 
This modification to the voter model allows us to model a variety of natures of individuals in society. For example, we can model strong-willed individuals by making the opinion evolution of the selected individual independent of the opinions of its neighbors. Similarly, we can model conformist and rebel individuals when the selected individual tends to adopt the opinion of the majority and minority respectively.  %We focus on the setting where the graph between the %individuals is a complete graph. This is justified in the presence of social media platforms like %Twitter and the abundance of publicly available poll results on most important issues. 

Use of social networks and other media outlets for political campaigning and advertising is on the rise. While opinions of individuals evolve organically over time, this evolution can be influenced by effective campaigning. Resource limitations like a fixed budget or limited manpower restrict the set of feasible influencing strategies and motivate the need to use  available resources efficiently. 

In political campaigning, the goal is to influence as many individuals as possible by the end of a fixed deadline.  Motivated by this, we focus on designing influencing strategies that maximize the number of individuals with a positive opinion at the end of a known and finite time horizon \cite{kandhway2014run}. The optimal influencing strategy is one that maximizes the number of individuals with a favorable opinion at the end of this time horizon. 

%In this work, our goal is to study how the nature of the optimal influencing strategy varies with the nature of individuals in society. 

Most political campaigns tend to ramp up their advertising as the elections gets closer. This strategy intuitively makes sense as once influenced by the advertisements towards the end of the time horizon, there is very little time for individuals to change their opinions. In this work, the goal is to understand if this advertising strategy is always optimal for our stylized model. %Our stylized variant of the voter model provides some insights into this matter. 

%In this work, our analytical results focus on the setting where the graph is a complete graph. %Via simulations, we show that the same trends extend other graph models including Erdos Renyi and Barabasi-Albert (power law graphs). 
The key takeaway from this work can be summarized as follows. For a society consisting primarily of strong-willed individuals who are unaffected by the opinion of their peers but are susceptible to external influence or rebels who tend to adopt the opinion of the minority, the optimal influencing strategy is to influence towards the end of the finite time horizon.  Contrary to this, if individuals are heavily influenced by their peers and are likely to adopt the opinion of the majority, in some cases, it is optimal to influence at the beginning of the time horizon. Intuitively, this is because increasing the fraction of individuals with a favorable opinion at the beginning of the finite time horizon has a cascading effect on opinions of the society as a whole. In addition, we also conclude that influencing at the beginning of the time horizon tends to be more effective than influencing at the end only when individuals with a favourable opinion are less likely to change their opinion than individuals with an unfavourable opinion. %We make these notions precise in the following sections. 

\subsection{Related Work}

The preliminary version of this work \cite{MTNS} was limited to conformists individuals, i.e., those who tend to adopt the opinion of the majority of their neighbors and strong-willed individuals, i.e., those whose opinion dynamics are independent of the opinion of their neighbours. In this work, we also consider a third type of individuals, called rebels. These individuals tend to adopt the opinion of the minority of their neighbors. In addition, in \cite{MTNS}, an individual’s opinion evolution mechanism is time-invariant, whereas in this work, we allow individuals to change their behavior across time.  

Close to our setting, \cite{majmudar2015voter} focuses on the voter model and generalizes it to include external influences. The key takeaway is that the effect of external influences overpowers node-to-node interactions in driving the network to consensus in the long term. In \cite{yildiz2011discrete}, the focus is on studying the effect of stubborn agents, i.e., agents who influence others but do not change their opinion, on the opinion dynamics of the network. The authors also study the problem of optimal placement of these stubborn agents to maximize the effect on the network. 

Designing optimal influencing strategies has been the subject of study in many works including \cite{kandhway2014run, kotnis2017incentivized, kempe2005influential, eshghi2017spread}. Refer to \cite{eshghi2017spread} for a detailed survey of various works in this domain. Unlike our work, most of these works focus on the infinite time horizon setting. In \cite{kandhway2014run}, the focus is on characterizing the optimal influence strategy to maximize the spread of an epidemic in a network. In \cite{kotnis2017incentivized}, the focus is on minimizing the cost incurred by the influencer to reach a fixed fraction of nodes in the network. In \cite{kempe2005influential}, the authors propose a general model of influence propagation called the decreasing cascade model and analyze its performance with respect to maximizing the spread of an idea. In \cite{eshghi2017spread}, the focus is on designing optimal advertising strategies in the presence of multiple advertising channels.

A related body of work focuses on preventing the spread of disease/viruses in networks (refer to \cite{asano2008optimal, ledzewicz2011optimal, lashari2012optimal} and the references therein). Our work differs from this body of work since we focus on strategies to increase the spread of favorable opinion in the network.

In \cite{chierichetti2009rumor, boyd2005gossip, pittel1987spreading}, the focus is on analyzing the performance of various rumor spreading strategies. These works do not focus on finding the optimal strategies for information spread.  

Other opinion dynamic models studied in literature include the linear threshold model where each edge has a non-negative weight and a node becomes active if the sum of the weights corresponding to its active neighbors crosses a threshold. Another example is the majority model where, in each round, all nodes switch their opinion the majority opinion in their neighborhood. Unlike these two models where opinions evolve in a deterministic manner, in our model, the opinion evolution of nodes is a random process. In addition, game-based models like the best response and noisy best response model have also been studied.

\subsection{Organization}

The rest of this paper is organized as follows. In Section \ref{section:setting}, we formally define our opinion dynamics model. In Section \ref{section:prelims}, we discuss some preliminary results on stochastic approximation which are used in the subsequent analysis. In Section \ref{section:main_results}, we state and discuss our main results. The proofs of the results discussed in Section \ref{section:main_results} are provided in Section \ref{section:proofs}. We conclude the paper in Section \ref{sec:Conclusion}. Some additional results are presented in the appendix. 

\section{Setting}
\label{section:setting}

We consider a finite population of $M$ individuals, where each individual has a binary opinion (``Yes" or ``No") about a certain (fixed) topic of interest. The opinion of each individual evolves over time. We assume the presence of an external influencing agency with a limited budget which tries to shape the opinion of the population over time. 

We assume that time is slotted and label the individuals in the population from $\{1, \ldots, M \}$. We define random variables $\{ I_i(t) \}_{1 \leq i \leq M, t \geq 0}$ taking values in $\{0, 1\}$, where $I_i(t)$ denotes the opinion of the $i^{th}$ individual in time-slot $t$. Thus,
$$ I_i(t) = \begin{cases}
1 \ \ \mbox{if the opinion of} \ i^{th} \ \mbox{individual at time} \ t \mbox{ is \textit{Yes}} \\
0 \ \ \mbox{if the opinion of} \ i^{th} \ \mbox{individual at time} \ t \mbox{ is \textit{No}}. 
\end{cases} $$
Let
$
Y(t) = \sum\limits_{i=1}^M I_i(t) \text{ and } N(t) = M - Y(t)
$
denote the total number of people at time $t$ with the opinion ``Yes" and ``No" respectively.  

At each time $t$, an individual (say $i_t$) is chosen uniformly at random from the population of $M$ individuals. The opinion of the chosen individual $i_t$ evolves in time-slot $t$ according to a Markov process with the following transition probabilities. 
\begin{align*}
P(I_{i_t}(t+1) = 0 | I_{i_t}(t) = 1) & = p_t\\
P(I_{i_t}(t+1) = 1 | I_{i_t}(t) = 1) & = 1-p_t\\
P(I_{i_t}(t+1) = 1 | I_{i_t}(t) = 0) & = q_t\\
P(I_{i_t}(t+1) = 0 | I_{i_t}(t) = 0) & = 1-q_t.
\end{align*}

The values of $p_t$ and $q_t$ depend on whether the individual is being externally influenced in time-slot $t$ or not. If the chosen individual is being externally influenced in time-slot $t$, $p_t = \tilde{p}$ and $q_t = \tilde{q}$, else, their opinion evolves in one of the following three ways. 

\begin{itemize}
	\item[--] Type $S$ (Strong-Willed): In this case, the chosen individual is not affected by their peers, more specifically, $p_t = p$ and $q_t = q$. 
	\item[--] Type $C$ (Conformist): In this case, the probability of an individual changing their opinion increases with the fraction of the population holding the opposite opinion, more specifically,  $p_t = p \frac{N(t)}{M}$ and $q_t = q \frac{Y(t)}{M}.$ %In the above setup this is the $\lambda=0$ case.
	\item[--] Type $R$ (Rebel): In this case, the probability of an individual changing their opinion decreases with the fraction of the population holding the opposite opinion, more specifically, \\ $p_t = p \left(1-\frac{N(t)}{M}\right)$ and $q_t = q \left(1-\frac{Y(t)}{M}\right)$. 
\end{itemize}

Since an individual can also adopt a mixture of these approaches over time, we study two models which allow time-varying nature in individuals. 

\begin{enumerate}
\item[--]{Model I: Hybrid $S/C$} \\
In each time-slot without external influence, opinion evolution of the chosen individual is Type $S$ with probability $\lambda$ and Type $C$ otherwise, independent of all past choices. 

Note that at the two extreme values of $\lambda$, i.e., $\lambda = 0$ and $\lambda = 1$, the individuals are either only strong-willed or only conformists. These special cases were studied in \cite{MTNS} along with an extension to the case of strong-willed population with increasing adamancy with time. Some of the results in \cite{MTNS} can be obtained as special cases of results in this paper.
\item[--]{Model II: Hybrid $C/R$} \\
In each time-slot without external influence, opinion evolution of the chosen individual is Type $C$ with probability $\mu$ and Type $R$ otherwise, independent of all past choices. 
\end{enumerate}
We focus our attention on these two models since they are analytically tractable and lead to insightful results. The analysis for the Hybrid $S/R$ model is similar to that of Hybrid $S/C$ and can be studied using the same tools. \color{black}

For both these models, we focus on the evolution of opinion of the population over a time horizon of $T$ consecutive time-slots. The advertising agency can influence the opinion in at most $bT$ of the $T$ time-slots, where $0 < b \leq 1$. The goal of the advertising agency is to maximize the number of individuals holding the ``Yes" opinion at the end of the $T$ time-slots. The agency needs to decide when to exert influence in order to achieve this goal.
 
\section{Preliminaries}
 \label{section:prelims}
In this section, we introduce the mathematical framework used in the rest of this paper. 
The dynamics for random variable $N(t)$ (resp. $Y(t)$) are as follows.
\begin{equation} \label{Nos} 
N(t+1) = N(t) + \chi(t+1), 
\end{equation}
where, $\chi(t+1)$ is a random variable taking values in $\{-1, 0, 1\}$ denotes the change in the net opinion of the population at time $t+1$. Let $\mathcal{F}_{t}$ denote the $\sigma$-field generated by the random variables $\{\chi(1), \chi(2), \ldots, \chi(t)\}$. The evolution of opinion is governed by the random process $\chi(t)$. We have:
\begin{eqnarray} \label{BasicModel} 
&& P(\chi(t+1)=x | \mathcal{F}_t) \nonumber \\
&&=
\begin{cases}
 \delta_{N}(t)q_t  &\text{ for} \ x=-1 \\ 
 (1-\delta_{N}(t))p_t  &\text{ for} \ x=1 \\
 1 - p_t - \delta_{N}(t)(q_t - p_t) &\text{ for} \ x=0,
  \end{cases} 
\end{eqnarray}
where, $\delta_N(t) = \frac{N(t)}{M}$. Note that the $\chi(k)$'s are conditionally independent.  As metioned in Section \ref{section:setting}, the values of $p_t$ and $q_t$ depend on whether the individual is being externally influenced in time-slot $t$ or not. If the chosen individual is being externally influenced in time-slot $t$, $p_t = \tilde{p}$ and $q_t = \tilde{q}$, else, their opinion evolves in one of the three ways described in Section \ref{section:setting}. 
\begin{defi}[Optimal Strategy]\label{defoptimal}
	We call a strategy optimal if the influence according to that strategy results in a larger expected number of \lq \lq Yes" at the end of time $T$ than the expected number of \lq \lq Yes" at the end of time $T$ using any other influence strategy.
\end{defi}
If strategy $\st_1$ is better than strategy $\st_2$, it is denoted by $\st_1 \gg \st_2$. Note that, for any strategy, the number of time-slots that can be influenced is fixed. We want to find an optimal strategy, in the sense of definition~\ref{defoptimal}, that identifies the time-slots where the influence should be exerted. As we shall see in most cases, due to a monotone argument, it is sufficient to compare the strategies of influencing in the first $bT$ and the last $bT$ time-slots respectively. 
We use the following definition in the rest of the paper. 

\begin{defi} [Influence Strategies]\color{white}df\color{black}
	\label{defn:strategies}
	\begin{enumerate}
%	\item[--]	$\tp = 0, \tq = 1$
		\item[--] $\st_F$: The strategy to influence in the first $bT$ time-slots. 
		\item[--] $\st_L$: The strategy to influence in the last $bT$ time-slots. 
	\end{enumerate}
\end{defi}

\noindent Total number of people with opinion \lq \lq No" at time $t$ is given by $N(t) = N(0) + \sum_{k=1}^t \chi(k)$. It is reasonable to assume that the influencing agency exerts influence in a way that leads to maximizing the expected number of people with opinion “Yes” at the end of time $T$. Consider a family of random variables $\{ \chi^{(p, q)} \}$, parametrized by $p, q \in [0, 1]$, defined as follows:
	\begin{equation*} \label{chi}
	P(\chi^{(p, q)} = x) = \begin{cases}
	p (1-\delta) & \text{if} \ x=1 \\
	q \delta & \text{if} \ x=-1  \\
	1-(1-\delta) p- \delta q  & \text{if} \ x=0.
	\end{cases}
	\end{equation*}
	for $\delta \in [0, 1]$. Then for $p_1 > p_2$ and $q_1 < q_2$, $\chi^{(p_1, q_1)}$ stochastically dominates $\chi^{(p_2, q_2)}$. In view of this, we assume the following.
	
%	 In other words, it is obvious that for an external agency to influence towards getting their favoured outcome, they should change the parameters such that probabilities to move away from the favourable opinion is reduced and the probability to move towards the favourable opinion increases. 

\color{black}

\begin{assumption}[Rational Influence] \label{rationalinfluence}
	We assume that the external influence is such that $\tp < \tq$. 
\end{assumption}

To analyze the evolution of random variable $\delta_N(t)$, we use the theory of constant step-size stochastic approximation. A constant step-size stochastic approximation scheme is given by a recursion of the form: 
\begin{equation}
x(n+1) = x(n) + a \left[h(x(n)) + M(n+1)\right]. \label{SAscheme}
\end{equation}
for $a>0, n \geq 0, x \in \mathbb{R}^d$, such that:
\begin{enumerate}[(i)]
%	\item $\{a(n)\}$ is a positive step-size sequence satisfying 
%	\begin{equation*} \sum_na(n) = \infty, \ \sum_na(n)^2 < \infty. \label{step} \end{equation*}
	\item $h: \mathbb{R}^d \to \mathbb{R}^d$ is Lipschitz.
	\item $\{ M_n \}_{n \geq 0}$ is a square-integrable Martingale difference sequence with respect to a suitable filtration.
\end{enumerate}
In addition, we assume that $\{ \| x_n \|^2 \}$ are uniformly integrable and $\sup_n E[\| x_n \|^2]^{1/2} < \infty$. From Chapter $9$ of \cite{borkar2008stochastic}, we know that solution(s) of \eqref{SAscheme} remain \lq \lq close" to the solution(s) of the O.D.E. given by: $\dot{x}(t) = h(x(t))$. Note that because of constant step-size (as opposed to the traditional decreasing step-size), solution(s) of the recursion do not converge almost surely to solution(s) of the O.D.E.. However, it is possible to get reasonable high probability bounds on the difference between solutions of the recursion and that of the O.D.E.. The idea is similar to that of standard decreasing step-size stochastic approximation and involves comparing the trajectories of the solutions of the O.D.E. with the interpolated trajectories of \eqref{SAscheme}. Define $t(n) = na$ for $n \geq 0$. Define $\bar{x}(\cdot)$ by $\bar{x}(t(n))=x_n \ \forall n$ such that $\bar{x}(t)$ is defined on $[t(n), t(n+1)$) by the following linear interpolation:
$$ \bar{x}(t) = x_n + (x_{n+1}-x_n)\frac{t-t(n)}{t(n+1)-t(n)}.$$
Let $x^s(\cdot)$ denote the solution of the O.D.E. for $t \geq s$, with $x^s(s)=\bar{x}(s)$. Thus, for $t \geq s$, $ x^s(t) = \bar{x}(s) + \int_s^t h(x^s(r))dr $.
From Lemma 1, Chapter 9 \cite{borkar2008stochastic}, we know that $E\left[\sup_{t \in [0, T]} \| \bar{x}(s+t) - x^s(s+t) \|^2\right] = O(a)$ for any $T>0$. In fact, if the O.D.E. has a globally stable equilibrium point $x^*$, Theorem 3, Chapter 9 \cite{borkar2008stochastic} implies that
$$ \lim\sup_{n \to \infty} P(\|x_n-x^*\| > \epsilon) = O(a).$$
While the discussion in \cite{borkar2008stochastic} is for the case when $\{ M_{n} \}$ is i.i.d. with Gaussian sequence, it is not difficult to see that the results hold for Martingale difference sequences.

\subsubsection{Stochastic Approximation Scheme for Our Model}
We rewrite \eqref{Nos} as a constant step-size stochastic approximation scheme for the fraction of \lq \lq No"s at time $t$.
\begin{eqnarray*} 
	\delta_N(t+1) &=& \delta_N(t) + \frac{1}{M} E[\chi(t+1) | \mathcal{F}_t] \\
	&&+ \frac{1}{M} [\chi(t+1) - E[\chi(t+1) | \mathcal{F}_t] ]. 
\end{eqnarray*}
Conditions (i), (ii) above and other conditions like boundedness of trajectories can be verified. From the theory of stochastic approximation, we know the solutions of the above system track the solutions of the O.D.E. $\dot{\delta}_{N}(t)=\frac{1}{M}E[\chi(t+1)|\mathcal{F}_{t}]$. Therefore, from~\eqref{BasicModel}, we get that the opinion dynamics is governed by the O.D.E.:
\begin{equation}\label{BasicODE}
M\dot{\delta}_N(t) = (1-\delta_N(t))p_t - \delta_N(t) q_t.
\end{equation}
It is sufficient to analyse the corresponding O.D.E. system to obtain the optimal strategies for the advertising agency. Figure~\ref{fig:tracking} shows that the O.D.E. corresponding to the constant step-size stochastic approximation scheme tracks the difference equation well throughout. %In the following sections, we discuss various cases and obtain optimal influencing strategies for the corresponding O.D.E.s.

\begin{figure}[!htbp] 
	\centering
	\includegraphics[width=0.5\textwidth]{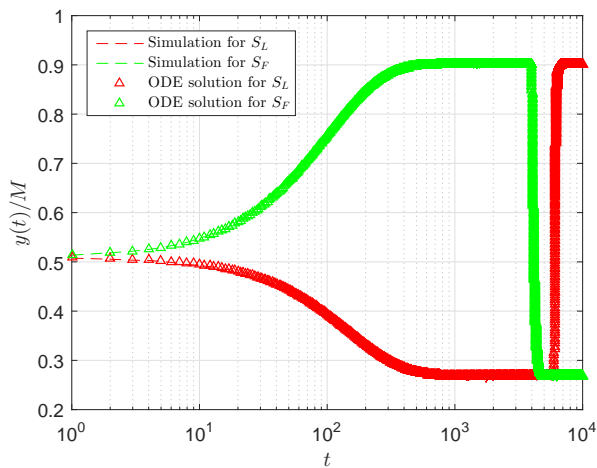}
		%OdeTracksHeterogenous_SC_p_GreaterThan_q_inTime_(0.8,0.4,0.1,0.9,1e4,1e2,100).png}
	\caption{Comparison between $\st_L$ and $\st_F$ for $M = 100$, $T = 10000$, $b=0.4$, $p=0.8$, $q=0.4$ and $\widetilde{p}=0.1, \widetilde{q}=0.9$ as a function of time for Model I with $\lambda = 0.5$. The O.D.E. corresponding to the constant stepsize stochastic approximation scheme tracks the difference equation well for all values of $t$.}
	\label{fig:tracking}
\end{figure}

\section{Main Results and discussion}
\label{section:main_results}

%We first discuss our results for Model I.
\subsection{Model I: Hybrid $S/C$}\label{sec:SC}
In this section, we characterize the optimal influence strategies for the Hybrid $S/C$ case. Recall that this means that in each time-slot when there is no external influence, the opinion of the randomly chosen individual evolves as Type S (strong-willed) with probability $\lambda$ and Type C (conformist) with probability $1-\lambda$. 

%\color{green}

Recall from Definition \ref{defn:strategies} that $\st_L$ and $\st_F$ are strategies to influence in the last and first $bT$ time-slots respectively.
\begin{theorem}
	\label{theorem:Model I}
	For Model I defined in Section \ref{section:setting} and under Assumption \ref{rationalinfluence}, if the advertiser has a budget of $bT$ time-slots for $0 < b < 1$, 
	\begin{enumerate}
		\item[(i)] If $p=q$, then, $\forall \lambda \in (0,1]$, $\st_L$ is optimal and for $\lambda=0$, all strategies perform equally well.
		\item[(ii)] If $p>q$, then, $\forall \lambda \in [0,1]$ and for $T = \omega(M)$ or $T=o(M)$, $\st_L$ is the optimal strategy.
		\item[(iii)]  If $p<q$, 
		\begin{enumerate}
		\item for $T=\omega (M)$, $\forall \lambda \in (0,1]$, $\st_F$ is optimal and for $\lambda=0$, all strategies perform equally well; 
		\item for $T=o (M)$, $\exists \lambda^{*} \in [0, 1]$  such that $\st_F$ is optimal and $\st_L$ is strictly sub-optimal when $\lambda < \lambda^*$, $\st_L$ is optimal when $\lambda > \lambda^*$, and if $\lambda = \lambda^*$, both strategies $\st_F$ and $\st_L$ perform equally well. 
		\end{enumerate}
	\end{enumerate}
\end{theorem}

\color{black}
%In $p<q$ case, for $\lambda = \lambda^*$, $\st_F$ and $\st_L$ perform equally well. We expect all strategies to perform equally in this case.

To discuss Theorem \ref{theorem:Model I} from a qualitative perspective, we introduce a property called the stickiness of an opinion. For a given value of the fraction of the population holding the same opinion as that of the individual chosen in a time-slot, the stickiness of that opinion is the probability that the chosen individual does not change their opinion by the end of the time-slot. Under this definition, for Type $S$ individuals, the stickiness of the Yes opinion is $1-p$ and that of the No opinion is $1-q$. 

The main takeaway from Theorem \ref{theorem:Model I} is that it is strictly sub-optimal to influence at the end of the time horizon when the Yes opinion is more sticky than the No opinion and the probability of an individual being affected by their peers is high. In addition, in this case, it is, in fact, optimal to influence right at the beginning of the time horizon. One way to understand this phenomenon is as follows. When individuals are heavily influenced by their peers and people are less likely to flip from Yes to No than from No to Yes, influencing the population at the beginning leads to a cascading effect which outperforms the strategy of influencing people at the very end which minimizes the probability of them switching their opinion before the end of the time horizon. 

In all other cases studied in Theorem \ref{theorem:Model I}, it is optimal to influence at the end of the time horizon. It is worth noting that in all the cases discussed in Theorem \ref{theorem:Model I}, the nature of the optimal policy is independent of the state of the population at the beginning of the time horizon.

 We now present some simulation results to illustrate that the performance of our system mirrors the trends obtained by solving the corresponding O.D.E. for this setting. %Due to space constraints, we only show results for the case when $p>q$ and $p<q$. The reuslts for the case when $p = q$ are consistent with the results obtained in Theorem \ref{theorem:Model I}.

%\begin{figure}[!htbp] 
%	\centering
%	\includegraphics[width=0.4\textwidth]{OdeTracksHeterogenous_SC_p_GreaterThan_q_inTime_(0.8,0.4,0.1,0.9,1e4,1e2,100).png}
%	\caption{Dynamics of simulated process remains close to the O.D.E. dynamics.}
%	\label{fig:tracking}
%\end{figure}

\begin{figure}[!htbp] 
	\centering
	\includegraphics[width=0.5\textwidth]{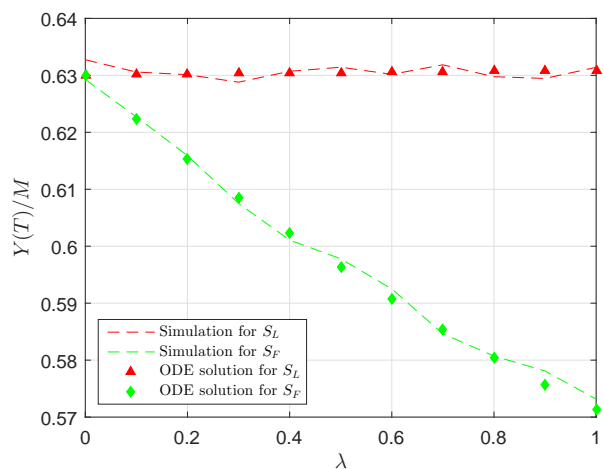}
	\caption{Comparison between $\st_L$ and $\st_F$ for $M = 10000$, $T = 10000$, $b=0.4$, $p=q=0.5$ and $\widetilde{p}=0.1, \widetilde{q}=0.9$ for different values of $\lambda$. The  $\st_L$ strategy outperforms the $\st_F$ strategy for all values of $\lambda$.}
	\label{figure:optimality_p=q}
\end{figure}
%
%In Figure \ref{figure:optimality_p=q}, we compare the performance of the $\st_L$ and $\st_F$ strategies for the setting where $p=q$. The O.D.E. solution is close to the simulated performance of the system. As discussed in Theorem \ref{theorem:Model I}, in this case, the $\st_L$ strategy outperforms the $\st_F$ strategy for all values of $\lambda$.
%
%
\begin{figure}[!htbp] 
	\centering
	\includegraphics[width=0.5\textwidth]{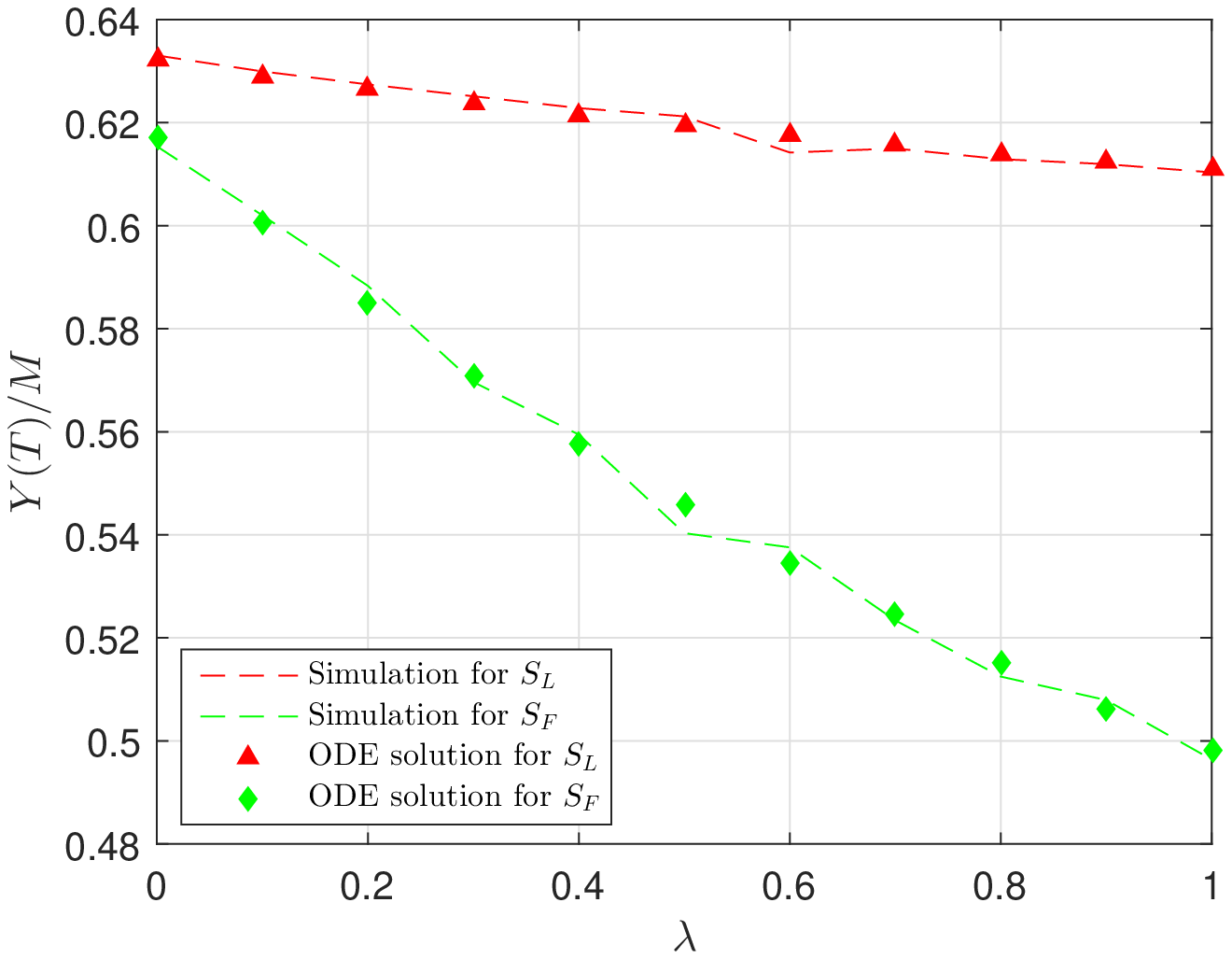}
	\caption{Comparison between $\st_L$ and $\st_F$ for $M = 1000$, $T = 1000$, $b=0.4$,  $p=0.8$, $q=0.4$ and $\widetilde{p}=0.1, \widetilde{q}=0.9$ for different values of $\lambda$. The  $\st_L$ strategy outperforms the $\st_F$ strategy for all values of $\lambda$.}
	\label{figure:optimality_p>q_1}
\end{figure}

\begin{figure}[ht]
	\centering
	\includegraphics[width=0.5\textwidth]{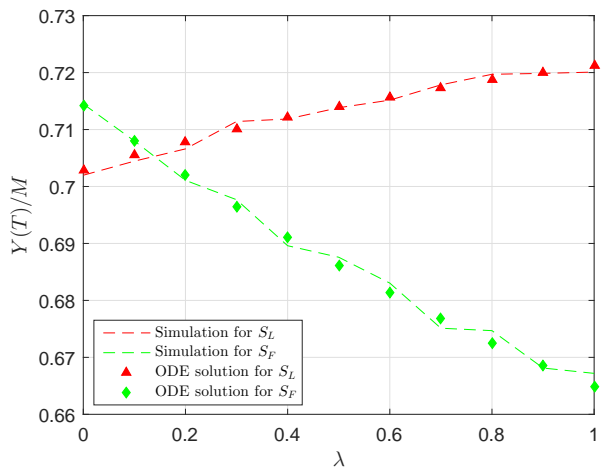}
	\caption{Comparison between $\st_L$ and $\st_F$ for $M = 1000$, $T = 100000$, $b=0.4$, $p=0.4, q=0.8$ and $\widetilde{p}=0.1, \widetilde{q}=0.9$. The O.D.E. solution is close to the simulated performance of the system. As discussed in Theorem \ref{theorem:Model I}, in this case, the $\st_L$ strategy outperforms the $\st_F$ strategy for values of $\lambda$ below a threshold and the $\st_F$ strategy outperforms the $\st_L$ strategy for values of $\lambda$ over the threshold.}
	\label{figure:crossover_p<q}
\end{figure}

In Figure \ref{figure:crossover_p<q}, we compare the performance of the $\st_L$ and $\st_F$ strategies for the setting where $p<q$. The O.D.E. solution is close to the simulated performance of the system. As discussed in Theorem \ref{theorem:Model I}, in this case, the $\st_L$ strategy outperforms the $\st_F$ strategy for values of $\lambda$ below a threshold and the $\st_F$ strategy outperforms the $\st_L$ strategy for values of $\lambda$ over the threshold.

%Figure~\ref{fig:crossover_p<q} demonstrates the existence of $\lambda^*$ where the the optimal strategy changes for the two regimes of $T/M$. Performance of a strategy which influences randomly chosen slots is also plotted for reference.
%
%
%
%\begin{figure}[ht]
%	\centering
%	\includegraphics[width=0.4\textwidth]{CrossOverLambdaDecrease.png}
%	%\includegraphics[width=0.4\textwidth]{Crossover Lambda as a function of q-p.png}
%	\caption{Crossover $\lambda$ for $p=q$.}
%	\label{fig:crossover_T=M}
%\end{figure}
%
%\begin{figure}[ht]
%	\centering
%	\includegraphics[width=0.4\textwidth]{Crossover Lambda as a function of q-p.png}
%	\caption{$\lambda^*$ as a function of $q-p$.}
%	\label{fig:crossover_q-p}
%\end{figure}

Note that some of the results in Theorem \ref{theorem:Model I} are restricted to the case when $T=\omega(M)$ and $T = o(M)$ for analytical tractability. We now present a result which holds for all $T$ for the special case when a individual does not change their opinion from Yes to No without external influence and the external influence is perfect, i.e., $\tilde{p} = 0$, and $\tilde{q} =1$. While this is a very limited case, our motivation behind discussing the result is to show that similar trends hold for general $T$.

\begin{proposition} \label{PS1}
	For Model I defined in Section \ref{section:setting}, if the advertiser has a budget of $bT$ time-slots for $0 < b < 1$, if $p=0 \text{ and } q>0$, and $\tilde{p} = 0$, and $\tilde{q} =1$, then, $\forall \lambda \in [0, 1)$ the strategy to influence in the first $bT$ slots is optimal and is strictly better than the strategy influencing in the last $bT$ slots. %The gap between the two is largest at $\lambda=0$ and shrinks to $0 $ as $\lambda$ tends to 1. The gap between the two strategies at $\lambda=0$ is non-monotonic in $q$ and can be maximized. 
\end{proposition}

\subsection{Model II: Hybrid $C/R$} \label{sec:CR}
In this section, we present our results for Model II. Recall that this means that in each time-slot when there is no external influence, the opinion of the randomly chosen individual evolves as Type C (conformist) with probability $\mu$ and Type R (rebel) with probability $1-\mu$.

%\color{blue}
%The following result holds only for for $T=o(M)$ or $T=\omega (M)$, . 
\begin{theorem}
	\label{theorem:Model II}
	For Model II defined in Section \ref{section:setting} and under Assumption \ref{rationalinfluence}, if the advertiser has a budget of $bT$ time-slots for $0 < b < 1$, 
	\begin{enumerate}
		\item[(i)] For $\mu < 1/2$, $\st_L \gg \st_F$ for $T=o(M)$ and $T=\omega (M)$.
		\item[(ii)] For  $\mu > 1/2$, 
		\begin{enumerate}
			\item If $T = o(M)$, $\st_F \gg \st_L$ if $p < \frac{\delta_N(0)^2}{1-\delta_N(0)^2}q$ and $\mu>\frac{\delta_{N}^{2}(0)}{2\delta_{N}^{2}(0)-p/(p+q)}$, and $\st_L \gg \st_F$ otherwise. 
			\item If $T = \omega(M)$, $\st_L \gg \st_F$.
		\end{enumerate}
		\item[(iii)] For $\mu > 1/2$ and $p = \frac{\delta_N(0)^2}{1-\delta_N(0)^2}q$, for $T = o(M)$ or $T = \omega(M)$, all strategies perform equally well.
	\end{enumerate}
\end{theorem}

We conclude that in the Hybrid $C$/$R$ setting, the strategy to influence at the beginning of the time-frame outperforms the strategy to influence at the end in a very limited case. This happens only when the time horizon is small, i.e., at most a vanishing fraction of the individuals change their opinion, and the chosen individual in a time-slot is more likely to conform than rebel, and the stickness of the Yes opinion is above a threshold which is a function of the initial state of the population and the stickness of the No opinion. 

%It is straightforward to see that influencing at the beginning is sub-optimal when people are more likely to swim against the tide since any headstart we get in terms of having a large number of people with the Yes option gets undone as the populations changes their opinion  

We note that unlike Model I, in this case, in addition to the stickiness of the two opinions, the initial state of the population determines which of the two strategies, namely, influencing right at the beginning of the time horizon and influencing right at the end of the time horizon performs better.

\begin{figure}[!htbp]
	\centering
	\includegraphics[width=0.5\textwidth]{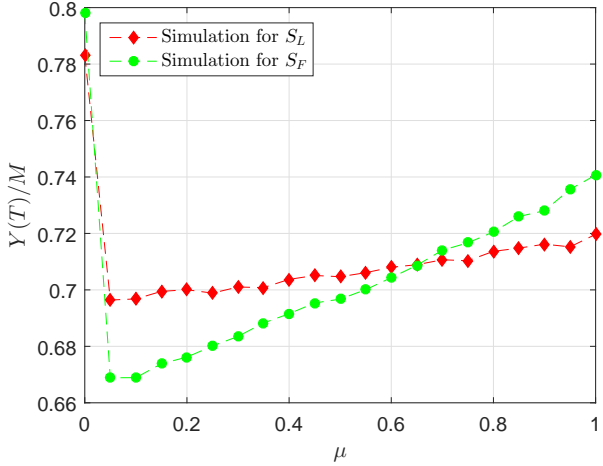}
	\caption{Comparison between $\st_L$ and $\st_F$ for $M = 1000$, $T = 1000$, $b=0.4$, $p=0.2, q=0.9$ and $\widetilde{p}=0.1, \widetilde{q}=0.9$. In this case, similar to that in Theorem \ref{theorem:Model II} (ii), the $\st_L$ strategy outperforms the $\st_F$ strategy for values of $\mu$ below a threshold and the $\st_F$ strategy outperforms the $\st_L$ strategy for values of $\mu$ over the threshold.}
	\label{fig:mixture}
\end{figure}

%%%%%%%%%%%%%%%%%%%%%%%%%%%%%%%%%%%%%%%%%%%%%%%%%%%%%%%%%%%%%%%%%%%%%%%%%%%%%%%%%%%%%%%%%%%%%%%%%%%%%%%%%%%%%%%%%%%%%%%%%%%%%%%%%%%%%%%%%%%%%%%%%%%%%%%%%%%%%%%%%%%%%%%%%%%%%%%%%%%%%%%%%%%%%%%%%%%%%%%%%%%%%%%%%%%%%%%%%%%%%%%%%%%%%%%%%%%%%%%%%%%%%%%%%%%%%%%%%%%%%%%%%%%%%%%%%%%%%%%%%%%%%%%%%%%%%%%%%%%%%%%%%%%%%%%%%%%%%%%%%%%%%%%%%%%%%%%%%%%%%%%%%%%%%%%%%%%%%%%%%%%%%%%%%%%%%%%%%%%%%%%%%%%%%%%%%%%%%%%%%%%%%%%%%%%%%%%%%%%%%%%%%%%%%%%%%%%%%%%%%%%%%%%%%%%%%%%%%%%%%%%%%%%%%%%%%%%%%%%%%%%%%%%%%%%%%%%%%%%%%%%%%%%%%%%%%%%%%%%%%%%%%%%%%%%%%%%%%%%%%%%%%%%%%%%%%%%%%%%%%%%%%%%%%%%%%%%%%%%%%%%%%%%%%%%%%%%%%%%%%%%%%%%%%%%%%%%%%%%%%%%%%%%%%%%%%%%%%%%%%%%%%%%%%%%%%%%%%%%%

\section{Proofs}
\label{section:proofs}

\subsection{Proof of Theorem \ref{theorem:Model I}(i)}

We consider the strategies $\st_F$ and $\st_L$ of influencing the first $bT$ time-slots and the last $bT$ time-slots respectively.  %Comparing the expected fraction of people with opinion \lq \lq No" obtained at the end of time $T$ by both strategies yields the following. 

\begin{proposition} \label{equalpq}
	For Model I defined in Section \ref{section:setting}, if the advertiser has a budget of $bT$ time-slots for $0 < b < 1$,  and $p=q$. Then, for all values of the population split $\lambda \in [0,1]$, the strategy of influencing in the last $bT$ time-slots is strictly better than the strategy of influencing in the first $bT$ time-slots, i.e. $\st_L \gg \st_F$.
\end{proposition}

\begin{proof}
	The expected fraction of \lq \lq No"s at the end of time $T$ under the strategies $\st_F$ and $\st_L$ are denoted by $\delta_N^{\st_F}(T)$ and $\delta_N^{\st_L}(T)$ respectively.
	
	We split the solution of the differential equation into two time-periods $[0,bT]$ and $[bT,T]$ depending on the presence or absence of influence corresponding to each strategy. For strategy $\st_F$ we have:
	\begin{itemize}
		\item $0\leq t\leq bT$: The differential equation is given by
		$$ M \dot{\delta}_{N}(t)=-\delta_{N}(t)[\widetilde{p}+\widetilde{q}]+\widetilde{p}. $$
		This can be solved to yield
		\begin{equation*}
		\delta_{N}(bT)=\frac{\widetilde{p}}{\widetilde{p}+\widetilde{q}}+\left(\delta_{N}(0)-\frac{\widetilde{p}}{\widetilde{p}+\widetilde{q}}\right) e^{-b T \frac{\widetilde{p}+\widetilde{q}}{M}}.
		\end{equation*}
		\item $bT\leq t\leq T$: The corresponding differential equation is solved by integrating between $bT$ and $T$. After simplification, this gives us: 
		\begin{eqnarray}
		\delta^{\st_F}_{N}(T)&=&\frac{1}{2}+\left(\frac{\widetilde{p}}{\widetilde{p}+\widetilde{q}}-\frac{1}{2}\right) e^{\frac{-2\lambda}{M}pT(1-b)} \nonumber \\
		&+& \left(\delta_{N}(0)-\frac{\widetilde{p}}{\widetilde{p}+\widetilde{q}}\right)e^{-\frac{T}{M}(2\lambda p(1-b)+b(\widetilde{p}+\widetilde{q}))}. \nonumber \\
		\label{F_equalpq}
		\end{eqnarray}
		
	\end{itemize}
	Similarly for the strategy $\st_L$, we obtain:
	\begin{eqnarray} 
	\delta^{\st_L}_{N}(T)&=& \frac{\widetilde{p}}{\widetilde{p}+\widetilde{q}} +\left(\frac{1}{2}-\frac{\widetilde{p}}{\widetilde{p}+\widetilde{q}}\right) e^{-bT \frac{\widetilde{p}+\widetilde{q}}{M}} \nonumber \\
	&+& \left(\delta_{N}(0)-\frac{1}{2}\right) e^{-\frac{T}{M}(2\lambda p(1-b)+b(\widetilde{p}+\widetilde{q}))}. \nonumber \\
	\label{L_equalpq}
	\end{eqnarray} 
	Comparing the expected fraction of number of people with \lq\lq No" at the end of time $T$ of $\st_F$ and $\st_L$, we get:
	
	\begin{eqnarray}
	\delta_{N}^{\st_F}(T)-\delta_{N}^{\st_L}(T)=
	\left(\frac{1}{2}-\frac{\widetilde{p}}{\widetilde{p}+\widetilde{q}}\right) \left(1-e^{-b T\frac{\widetilde{p}+\widetilde{q}}{M}}\right) \nonumber \\
	\left(1-e^{\frac{-2\lambda}{M}(pT(1-b)}\right). \label{diff_equalpq}
	\end{eqnarray}
	This is always positive whenever $\widetilde{p}\leq\widetilde{q}$. Thus, whenever $p=q$ and $\widetilde{p}\leq\widetilde{q}$, the expression in equation~\eqref{diff_equalpq} is a product of positive terms and is therefore strictly positive for all $\lambda, p, q, M \text{ and } T$. 
\end{proof}

\begin{remark} It is worthwhile to note that if the influencing agency exerts perfect influence i.e., $\widetilde{p}=0$ and  $\widetilde{q}=1$, we obtain $\delta_{N}(bT)=\delta_{N}(0)  e^{-\frac{bT}{N}}$, an exponential decay in the proportion of people who have a \lq \lq No" opinion. When $\delta_{N}(t)$ is very small, the differential equation governing $\delta_{N}(t)$ can be approximated by $M \dot{\delta}_{N}(t)=p\lambda$ which indicates a linear increase in $\delta_{N}(t)$. 
\end{remark}

We have only shown that $\st_L \gg \st_F$. However, it turns out this is sufficient to conclude that under the conditions of Proposition~\ref{equalpq}, $\st_L$ is in fact the optimal strategy. 

\begin{proof}[Proof of Theorem \ref{theorem:Model I}(i)]
	Consider a strategy $\st$ and divide the time interval $[0, T]$ in three parts: $[0, t-1)$ be the first $t \geq 0$ slots until a pair of slots is encountered where influence is followed by no influence, $[t, t+1]$ be the first pair of slots (starting from $t=0$) such that $t$ is an influenced time-slot and $t+1$ is not, and finally $(t+1, T]$. We consider a strategy $\st^\prime$ that differs from $\st$ only in the slots $[t, t+1]$, where the influence and non-influence is swapped. It is immediate from \eqref{F_equalpq} and \eqref{L_equalpq} that $\st^\prime \gg \st$. Inductively, we conclude that $\st_L$ is optimal in the sense of Definition~\ref{defoptimal}. 
\end{proof}

\color{black}
\subsection{Proof of Theorem \ref{theorem:Model I}(ii)}
In the absence of any external influence, the differential equation is given by: 
%	\begin{equation}
%	\ \dot{\delta}_N(t) = - \frac{\delta_N^2(t)}{M} (1-\lambda)(p-q) + \frac{\delta_N(t)}{M} [(1-\lambda)(p-q)-\lambda (p+q)] + \frac{\lambda p}{M}   \label{p>q_noinfluence}
%	\end{equation}
		\begin{equation}
	\ \dot{\delta}_N(t) = -f(\delta_N(t))  \label{p>q_noinfluence}
	\end{equation}
where, $f(\delta_N(t))=  \frac{\delta_N^2(t)}{M} (1-\lambda)(p-q) - \frac{\delta_N(t)}{M} [(1-\lambda)(p-q)-\lambda (p+q)] - \frac{\lambda p}{M} $ is a quadratic in $\delta_N(t)$ with two distinct roots for all values of $\lambda, p$ and $q$. We denote the roots of $f(\delta_N(t))$ by $A_1$ and $A_2$. Then,
$A_{1}=\frac{1}{2}\left(1+\frac{\sqrt{M^{2}\Delta}-\lambda(p+q)}{(1-\lambda)(p-q)}\right)$ and $A_{2}=\frac{1}{2}\left(1-\frac{\sqrt{M^{2}\Delta}+\lambda(p+q)}{(1-\lambda)(p-q)}\right)$, where $\Delta > 0$ denotes the discriminant of the quadratic. Note that $A_1>A_2$. We get:
\begin{eqnarray} \label{ODE}
\dot{\delta}_N(t) = - \frac{(1-\lambda)(p-q)}{M} (\delta_N(t) - A_1)(\delta_N(t) - A_2) 
\end{eqnarray}
%Or,
%\begin{eqnarray} \label{ODE}
%\frac{\dot{\delta}_N(t)}{ (\delta_N(t) - A_1)} - \frac{d \delta_N(t)}{(\delta_N(t) - A_2) }  = - \frac{(1-\lambda)(p-q)(A_1-A_2)}{M} dt
%\end{eqnarray}
Let $L=(1-\lambda)\frac{p-q}{M}(A_{1}-A_{2})$ and $D_{1}$, $D_{2}$ denote the expression $1-\left(\frac{\rho-A_1}{\rho-A_2}\right)e^{-LT(1-b))}$ evaluated at $\rho=\delta_{N}(0)$ and $\rho=\delta_{N}(bT)$ respectively. For the strategy $\st_F$, for $0 \leq t < bT$, we have:
$$ \delta_N(bT) = \frac{\tp}{\tp+\tq} + \left( \delta_N(0) - \frac{\tp}{\tp+\tq} \right)e^{-bT\frac{\tp+\tq}{M}} $$
and for $bT \leq t \leq T$, from \eqref{ODE}, we get:
\begin{eqnarray*} 
\log \left( \frac{ \delta_N(T) - A_1}{\delta_N(T) - A_2} \right) - \log \left( \frac{ \delta_N(bT) - A_1}{\delta_N(bT) - A_2} \right) \\
= - \frac{(1-\lambda)(p-q)(A_1-A_2)}{M} (1-b)T
\end{eqnarray*}
%That is,
%\begin{eqnarray*}
%\frac{ \delta_N(T) - A_1}{\delta_N(T) - A_2 }  &=& \left( \frac{ \delta_N(bT) - A_1}{\delta_N(bT) - A_2 } \right) \\
%&& \times \exp \left( - \frac{(1-\lambda)(p-q)(A_1-A_2)}{M} (1-b)T \right) \\
%&=& \left( \frac{ \delta_N(bT) - A_1}{\delta_N(bT) - A_2 } \right)  e^{- L(1-b)T }
%\end{eqnarray*}
That is,
\begin{eqnarray*}
\frac{ \delta_N(T) - A_1}{\delta_N(T) - A_2 }  =  \left( \frac{ \delta_N(bT) - A_1}{\delta_N(bT) - A_2 } \right)  e^{- L(1-b)T }
\end{eqnarray*}
Solving this, we get:
\begin{equation} \label{p>q_stF}
\delta_{N}^{\st_F}(T) = A_{2}+\frac{A_{1}-A_{2}}{D_{2}}  
\end{equation}
Similarly,
\begin{eqnarray} \label{p>q_stL}
\delta_N^{S_L}(T) &=& \frac{\tp}{\tp+\tq} \left( 1 - e^{-bT \frac{\tp+\tq}{M}} \right)\nonumber \\
&& +(A_1-A_2) \left( \frac{e^{-bT \frac{\tp+\tq}{M}}}{D_1} \right) + A_2 e^{-bT \frac{\tp+\tq}{M}} \nonumber \\
\end{eqnarray}
From ~\eqref{p>q_stF} and \eqref{p>q_stL}, we get:
\begin{eqnarray*}
&& \delta_N^{S_L}(T) - \delta_N^{S_F}(T)\\
 &=& \frac{\tp}{\tp+\tq} (1-e^{-bT \frac{\tp+\tq}{M}}) + \frac{A_1-A_2}{D_1 D_2} \left( D_2 e^{-bT \frac{\tp+\tq}{M}} - D_1 \right)\\ 
&=& \frac{\tp}{\tp+\tq} (1-e^{-bT \frac{\tp+\tq}{M}}) - (1-e^{-bT \frac{\tp+\tq}{M}})(1-e^{-LT(1-b)}) \\
&& + e^{-bT \frac{\tp+\tq}{M}}e^{-LT(1-b)} \left( 1- \frac{\delta_N(bT) - A_1}{\delta_N(bT) - A_2} \right) \\
&& - e^{-LT(1-b)} \left( 1- \frac{\delta_N(0) - A_1}{\delta_N(0) - A_2} \right) \\
&=& (1-e^{-bT \frac{\tp+\tq}{M}}) \\
&&\times\left[\left(\frac{\tp}{\tp+\tq}-A_2\right)\left(1- e^{-LT(1-b)} \frac{(A_1-A_2)^2}{D_1D_2D_3}\right) \right] \\
&& - (1-e^{-bT \frac{\tp+\tq}{M}})\left[ \frac{A_1-A_2}{D_1D_2} \left(1-e^{-LT(1-b)} \right) \right]
\end{eqnarray*}
where, $D_3 = (\delta_N(bT)-A_2)(\delta_N(0)-A_2)$.

\color{black}

\begin{remark}
	It can be shown that $D_{1},D_{2}$ and $D_{3}$ are always positive. Further, $A_{1}\in(0,1)$ and $A_{2}\in(-\infty,0)$ with both being decreasing functions in $\lambda$.
\end{remark}

Simulations indicate that in this case $\st_L \gg \st_F$. We prove this next for the cases when $T = o(M)$ and $T=\omega(M)$. \footnote{We say that $f(n)$ is $o(g(n))$ (or $\omega(h(n))$ resp.) if for any real constant $c > 0$, there exists an integer constant $n_0 \geq 1$ such that $f(n) < c g(n)$ (or $f(n) > c h(n)$ resp.) for every integer $n \geq n_0$.}
That is, whenever the influencing agency is either able to reach almost no one in the population or able to reach almost everyone, the strategy to influence in the last $bT$ slots is strictly better than the strategy to influence in the first $bT$ slots. 

\begin{lemma}
	For a system of $M$ individuals and a time horizon of $T$ time-slots, let $E$ be the event that each individual is influenced at least once in the time horizon of $T$ time-slots and $F$ be the event that the number of unique influenced individuals is $o(M)$. Then we have that,
	\begin{align*}
	\text{If } & T = \omega(M), \ \lim_{M \to \infty} P(E)  = 1 \\
	\text{ and if } &T = o(M),  P(F)  = 1.
	\end{align*}
\end{lemma}

\begin{proof}
	If $T = \omega(M)$, the probability that an individual is not influenced is $\left( 1-\frac{1}{M}\right)^{bT}$. By the union bound, 
	\begin{equation*}
	P(E) \geq 1 - M \left( 1-\frac{1}{M}\right)^{bT}
	\implies  \lim_{M \to \infty}  P(E)  = 1.
	\end{equation*}
	
	If $T = o(M)$, it trivially follows that the number of unique influenced individuals is upper bounded by $bT = o(M)$.
\end{proof}

\footnotetext{We write $f(x) \approx g(x)$ if $f(x) = g(x) + O(x)$}

%\begin{theorem}\label{p>q}
%	Consider the opinion dynamics in \eqref{ModelODE} with assumption~(\ref{rationalinfluence}) and $p>q$. Then, for all values of the population split $\lambda \in [0,1]$ and for $T = \omega(M)$ and $T=o(M)$, $\st_L$ is the optimal strategy. 
%\end{theorem}

\begin{proof}[Proof of Theorem \ref{theorem:Model I}(ii)]
	Note that it is enough to show that $\st_L \gg \st_F$ since an argument on the lines of the proof of Part (i) gives the optimality result. We consider the two regimes separately.
	\begin{itemize}
		\item $T=\omega(M)$: In this case, we have $TL(1-b)=(1-\lambda)(p-q)(A_{1}-A_{2})\frac{T(1-b)}{M}$. 
		We can simplify the above expression by noting that $(A_{1}-A_{2})(1-\lambda)(p-q)=\sqrt{M^{2}\Delta}=\sqrt{(\lambda)^{2}(p+q)^{2}+(1-\lambda^{2})(p-q)^{2}}$. Thus, $e^{-LT(1-b)} \to 0$. Using this, we get that $D_{1}\approx1,D_{2}\approx1$ and $D_{3}\approx -A_{2}(\delta_N(0)-A_2)$. The difference equation then collapses to $\delta_{N}^{\st_L}(T)-\delta_{N}^{\st_F}(T) \approx - A_2-(A_1-A_2)=-A_1$. Since this is always negative, we see that the strategy $\st_L$ outperforms the strategy $\st_F$. 
		%This means that when we are able to reach every individual through the advertising agency, the $\st_L$ strategy is better than the $\st_F$ strategy for all values of $\lambda$. 
		
		%\item $T=o(M)$: In this case, we have $e^{-LT(1-b)}=e^{-\sqrt{M^{2}\Delta}(1-b)x}=e^{-\mu x}$ where $x=\frac{T}{M}$. When $x \to 0$, $e^{-LT(1-b)} \approx 1-x\mu$. This gives us $D_{1} \approx D_{2} \text{ and } D_{3}\approx(\delta_{N}(0)-A_{2})^{2}$. \\
		%We can write $D_{1}D_{2}D_{3}=((A_{1}-A_{2})+x\mu(\delta_{N}(0)-A_{1}))^{2}$. The first term in the difference equation then simplifies to $-A_{2}(x\mu(\frac{2\delta_{N}(0)-(A_{1}+A_{2})}{A_{1}-A_{2}})-(\frac{2(\delta_{N}(0)-A_{1})}{A_{1}-A_{2}})(x\mu)^{2})$. The second term in the difference equation is $\frac{x\mu(A_{1}-A_{2}) (\delta_{N}(0)-A_{2})^{2}}{((A_{1}-A_{2})+x\mu(\delta_{N}(0)-A_{1}))^{2}}$. Define the quantity $\kappa=\frac{(A_{1}-A_{2})^{2}}{((A_{1}-A_{2})+x\mu(\delta_{N}(0)-A_{1}))^{2}}$. Then, the difference equation simplifies to $-A_{2}(1-(1-x\mu) \kappa)-\frac{x\mu(\delta_{N}(0)-A_{2})^{2}}{A_{1}-A_{2}} \kappa$ which after a series of tedious reduction steps becomes $\frac{x\mu\kappa}{A_{1}-A_{2}}(A_{1}A_{2}-(\delta_{N}(0))^{2}$ which is negative for all $\lambda$. This is because $A_{1}>0 \text{ but } A_{2}<0$. Thus we get that the $\st_L$ strategy is better than the $\st_F$ strategy even when the advertising agency is able to only reach a very small proportion of the population. 
		%\end{itemize}
		\item $T=o(M)$: In this case, we have $e^{-LT(1-b)}=e^{-\sqrt{M^{2}\Delta}(1-b)(T/M)}$. As $T/M \to 0$, $e^{-LT(1-b)} \approx 1- \frac{T}{M}\sqrt{M^{2}\Delta}(1-b)$. This gives us $D_{1} \approx D_{2} \text{ and } D_{3}\approx(\delta_{N}(0)-A_{2})^{2}$. By writing $D_{1}D_{2}D_{3}$ as $((A_{1}-A_{2})+\frac{T}{M}\sqrt{M^{2}\Delta}(1-b)(\delta_{N}(0)-A_{1}))^{2}$,the first and second terms in the difference equation simplify significantly. Let $\psi = \frac{(A_{1}-A_{2})^{2}}{((A_{1}-A_{2})+\frac{T}{M}\sqrt{M^{2}\Delta}(1-b)(\delta_{N}(0)-A_{1}))^{2}}$. Then, the difference equation simplifies to $\frac{\frac{T}{M}\sqrt{M^{2}\Delta}(1-b)\psi}{A_{1}-A_{2}}(A_{1}A_{2}-(\delta_{N}(0))^{2}$ which is negative for all $\lambda$. This is because $A_{1}>0 \text{ but } A_{2}<0$. Thus we get that the $\st_L$ strategy is better than the $\st_F$ strategy.
	\end{itemize}
\end{proof}
\vspace{-0.5in}

\subsection{Proof of Theorem \ref{theorem:Model I}(iii)}
\label{section:p<q}
In this case, in the absence of influence, the differential equation can be rewritten as
\begin{align}
\dot{\delta}_{N}(t)=&\frac{\delta_{N}^{2}(t)}{M}(1-\lambda)(q-p)-\frac{\delta_{N}(t)}{M}[(1-\lambda)(q-p) \nonumber \\
&+\lambda(p+q)]+\frac{\lambda p}{M}. \label{eq:8}
\end{align}

This quadratic equation also has two distinct roots $A_{1}^{\prime}$ and $A_{2}^{\prime}$ given by $A_{1}^{\prime}=\frac{1}{2}\left(1+\frac{\sqrt{N^{2}\Delta}+\lambda(p+q)}{(1-\lambda)(q-p)}\right)$ and $A_{2}^{\prime}=\frac{1}{2}\left(1+\frac{\lambda(p+q)-\sqrt{N^{2}\Delta}}{(1-\lambda)(q-p)}\right)$, where $\Delta=\frac{(p-q)^{2}(1-\lambda^{2})+\lambda^{2}(p+q)^{2}}{N^{2}}$ is the discriminant. Solving like before, we observe that: 
\begin{eqnarray}
\delta_{N}^{\st_F}(t)&=& A_{2}^{\prime}+\frac{A_{1}^{\prime}-A_{2}^{\prime}}{D_{2}^{\prime}} \\
\delta_{N}^{\st_L}(t)&=& \left(\frac{\widetilde{p}}{\widetilde{p}+\widetilde{q}}\right) \left(1-e^{-b  T  \frac{\widetilde{p}+\widetilde{q}}{M}}\right) \nonumber \\
&+&  A_{2}^{\prime} \left(e^{-b  T  \frac{\widetilde{p}+\widetilde{q}}{M}}\right)+ \left(A_{1}^{\prime}-A_{2}^{\prime}\right)  \left(\frac{e^{-b  T  \frac{\widetilde{p}+\widetilde{q}}{M}}}{D_{1}^{\prime}}\right). \nonumber
\end{eqnarray}
Thus we get: 
	\begin{align}\label{diff_lambda_unique}
	&\delta_{N}^{\st_L}(t)-\delta_{N}^{\st_F}(t) \nonumber \\
	&= \left(1-e^{-bT \frac{\widetilde{p}+\widetilde{q}}{M}}\right) \times \nonumber \\
	&  \left[\left(\frac{\widetilde{p}}{\widetilde{p}+\widetilde{q}}-A_{2}^{\prime}\right)\left(1-e^{L^{\prime}T(1-b)}\frac{(A_{1}^{\prime}-A_{2}^{\prime})^{2}}{D_{1}^{\prime}D_{2}^{\prime}D_{3}^{\prime}}\right)\right] \nonumber \\
	&  - \left(1-e^{-b  T  \frac{\widetilde{p}+\widetilde{q}}{M}}\right)\left[ \frac{A_{1}^{\prime}-A_{2}^{\prime}}{D_{1}^{\prime}D_{2}^{\prime}}\left(1-e^{L^{\prime}T(1-b)}\right)\right],
	\end{align}

where $L^{\prime}=(1-\lambda)\frac{p-q}{N}(A_{1}^{\prime}-A_{2}^{\prime})$ and $D_{1}^{\prime}$ and $D_{2}^{\prime}$ is the expression $\left[1-\frac{(x-A_{1}^{\prime})}{x-A_{2}^{\prime}}e^{L^{\prime}T(1-b)}\right]$ evaluated at $x=\delta_{N}(0)$ and $x=\delta_{N}(bT)$ respectively. Also, $D_{3}^{\prime}=(\delta_N(bT)-A_2^{\prime}) (\delta_N(0)-A_2^{\prime})$. 

Observe that since the equation in \eqref{diff_lambda_unique} takes opposite signs at extremal values of $\lambda$, there exists a $\lambda^*$, which we call the cross-over $\lambda$, such that for $\epsilon > 0$, for $\lambda$ in the neighbourhood $(\lambda^* - \epsilon, \lambda^* + \epsilon)$, $\st_F \gg \st_L$ for $\lambda < \lambda^*$ and $\st_L \gg \st_F$ for $\lambda > \lambda^*$. We show that for $T = o(M)$ or $\omega (M)$, the crossover $\lambda$ is unique.

\begin{remark}
	It can be shown that $A_{1}^{\prime} \in [1,\infty)$, $A_{2}^{\prime} \in \left[0,\frac{p}{q+p}\right]$ and both are increasing functions in $\lambda$. The variation of $D_{1},D_{2}\text{ and }D_{3}$ is much more complex and depends on $\delta_{N}(0)$. In particular, if $\delta_N(0)>(\frac{p}{q+p})  e^{\frac{bT}{M}}$, all three are positive.
\end{remark}

%\begin{theorem}[{Uniqueness of $\lambda^*$}] \label{p<q}
%	Consider the opinion dynamics in \eqref{ModelODE} with assumption~(\ref{rationalinfluence}) and $p<q$. Then, for all values of the population split $\lambda \in [0,1]$ and for $T = \omega(M)$ and $T=o(M)$, there exists a unique value of the population split ($\lambda^{*}$) such that $\st_F$ is optimal when $\lambda < \lambda^*$ and $\st_L$ is optimal when $\lambda > \lambda^*$.
%\end{theorem}

\begin{proof}[Proof of Theorem \ref{theorem:Model I}(iii)]
	We compare the strategies $\st_L$ and $\st_F$. An argument similar to that in the proof of Part (i) gives us the optimality of the strategies in different regimes of $\lambda$. Again, we divide the proof into two cases:
	\begin{itemize}
		\item $T=\omega(M)$: In this case, $TL^{\prime}(1-b)=(1-\lambda)(p-q)(A_{1}^{\prime}-A_{2}^{\prime})\frac{T(1-b)}{M}$. 
		We can simplify the above expression by noting that $(A_{1}^{\prime}-A_{2}^{\prime})(1-\lambda)(p-q)=\sqrt{M^{2}\Delta}=\sqrt{(\lambda)^{2}(p+q)^{2}+(1-\lambda^{2})(q-p)^{2}}$. We can write the difference equation as:
		\begin{small}
			\begin{eqnarray*}
				\delta_{N}^{\st_L}(T)-\delta_{N}^{\st_F}(T) &\approx & -A_{2}^{\prime}+A_{2}^{\prime}  e^{L^{\prime}T(1-b)}\frac{(A_{1}^{\prime}-A_{2}^{\prime})^{2}}{D_{1}^{\prime}D_{2}^{\prime}D_{3}^{\prime}} \nonumber \\
				&& -\frac{A_{1}^{\prime}-A_{2}^{\prime}}{D_{1}^{\prime}D_{2}^{\prime}}+\frac{A_{1}^{\prime}-A_{2}^{\prime}}{D_{1}^{\prime}D_{2}^{\prime}}  e^{L^{\prime}T(1-b)}) 
			\end{eqnarray*}
		\end{small}
		which can be simplified to yield:
		$\delta_{N}^{\st_L}(T)-\delta_{N}^{\st_F}(T) \approx 
		A_{2}^{\prime} \left[e^{-L^{\prime}T(1-b)} \left(1-\frac{A_{2}^{\prime}}{A_{1}^{\prime}}\right)-1\right]$ 
		but we know that  $\frac{A_{2}^{\prime}}{A_{1}^{\prime}}<1$ for all $\lambda$ except when $\lambda=0$ (since $\lambda = 0$ implies $A_{2}^{\prime}=0$). Thus the difference is always negative except when $\lambda=0$. Thus, in this case, $\lambda^*$ is $0$ and unique. 
		
		\item $T=o(M)$: In this case, we have $e^{L^{\prime}T(1-b)}=e^{\sqrt{M^{2}\Delta}(1-b)T/M}$. As $T/M \to 0$, $e^{L^{\prime}T(1-b)} \approx 1+\sqrt{M^{2}\Delta}(1-b) \frac{T}{M}$. This gives us $D_{1}^{\prime} \approx D_{2}^{\prime} \text{ and } D_{3}^{\prime}\approx(\delta_{N}(0)-A_{2}^{\prime})^{2}$. We can write $D_{1}^{\prime}D_{2}^{\prime}D_{3}^{\prime}=((A_{1}^{\prime}-A_{2}^{\prime})-x\mu(\delta_N(0)-A_{1}^{\prime}))^{2}$. %Then, the first term of the difference equation becomes $A_{2}^{\prime} (x\mu)\frac{2\delta_N(0)-(A_{2}^{\prime}+A_{1}^{\prime})}{A_{1}^{\prime}-A_{2}^{\prime}}$. We can proceed similarly (and tediously) to obtain the second term of the difference equation as $x\mu  \frac{(\delta_{N}(0)-A_{2}^{\prime})^{2}}{A_{1}^{\prime}-A_{2}^{\prime}}$. 
		Thus reducing the first and second terms of the difference equation appropriately and plugging in, the difference equation reduces to: 
		\begin{equation} \label{diff_p<q}
		\delta_{N}^{\st_L}(T)-\delta_{N}^{\st_F}(T) \approx \frac{x\mu}{A_{1}^{\prime}-A_{2}^{\prime}} [\delta_{N}^{2}(0)-A_{1}^{\prime}A_{2}^{\prime}]
		\end{equation} 
		$A_{1}^{\prime} \text{ and } A_{2}^{\prime}$ are both strictly increasing functions. $A_{1}^{\prime}A_{2}^{\prime}=0$ at $\lambda=0$ and increases with $\lambda$. So, the difference in~\eqref{diff_p<q} goes from positive to zero, becomes negative and then stays negative. Thus, $\lambda^{*}$ is unique.
	\end{itemize}
\end{proof}

\vspace{-0.4in}

\subsection{Proof of Proposition \ref{PS1}}
	
	\begin{proof}
		Note that this is the $p<q$ case. From the discussion in Section~\ref{section:p<q}, we get that $A_{2}^{\prime} = 0$. Also, $D_{1}^{\prime}>1$ and $D_{2}^{\prime}>1$. With this simplification, the difference in \eqref{diff_p<q} simplifies to: 
		\begin{equation}
		\delta_{N}^{\st_L}(t)-\delta_{N}^{\st_F}(t)=
		(1-e^{-\frac{bT}{M}})  \frac{e^{\frac{Tq(1-b)}{M})}-1}{(1-\lambda){D_{1}^{\prime}D_{2}^{\prime}}}
		\end{equation}
		which is positive except for $\lambda \in [0, 1)$ 
	\end{proof}
	The performance gap of both strategies shrinks to $0$ as $\lambda \to 1$. This is consistent with the result obtained in \cite{MTNS} (See Lemma 2 in \cite{MTNS}).
	
	%\vspace{-0.2in}

\subsection{Proof of Theorem \ref{theorem:Model II}}

The proof follows using the same tools used in the proof of Theorem \ref{theorem:Model I}. We omit the details due to lack of space. 

\section{Conclusions and Future Work} \label{sec:Conclusion}

In this work, we proposed a variant of the voter model which can be used to model variation in the nature of the individuals in society. We evaluate the performance of campaigning strategies as a function of the nature of individuals when the goal is to maximize the fraction of individuals with a favorable opinion at the end of a known finite time-horizon.

We conclude that if individuals are mostly unaffected by the opinion of their peers or tend to go against the majority opinion, influencing at the end of the finite time-horizon is optimal. In the case where individuals are affected by the opinion of their peers and tend to adopt the opinion of the majority, influencing at the end of the finite time-horizon can be strictly sub-optimal if an individual with a positive opinion is not very likely to change their mind when compared to the probability of an individual with a negative opinion changing their mind. 

Possible extensions of this work include modeling the connections between individuals in the society using a graph such that individuals susceptible to being influenced by others are only influenced by their neighbors in this graph.

\section*{Appendix: Martingale Concentration} \label{sec:App}
In this section, we use Concentration inequalities for Martingales to obtain a concentration results for fraction of people with opinion \lq \lq No", 
%at the end of sufficiently long time period $T$, 
for the $p=q$ case and the influence is in the last $bT$ time-slots for Model I. Similar arguments will give the corresponding result for the strategy to influence in the first $bT$ time-slots. We show that $\delta_N(T)$ is close to the solution of the corresponding O.D.E. in \eqref{BasicODE} at $T$. This justifies using the O.D.E. solutions to arrive at the optimal strategy for the external influencing agency for the discrete time model. 
%Throughout this section, $x \stackrel{y}{\approx} z$ means $|x-z| \leq y$. We use the following notation: $\alpha = \frac{\p+\q}{N}$, $\tilde{\alpha} = \frac{\pt +\qt}{N}$ and $\widehat{\alpha}_t = \frac{\widehat{p}_t+\widehat{q}_t}{N}$. Note that for the first $bT$ influence strategy:
%\begin{equation} \label{al}
%\widehat{\alpha}_t = \begin{cases} \tilde{\alpha} & \mbox{for} \ t \in [0, bT] \\ 
%\alpha & \mbox{for} \ t \in (bT, T]. \end{cases}
%\end{equation}
%\begin{proposition} 
%Given $\epsilon > 0$, there exists $\delta = \delta(\epsilon, T)$ such that:
%$$P\left(\Big|\frac{r(T)}{N} - M(T) \Big| < \epsilon \right) \geq 1-\delta,$$
%where, $M(T) = \frac{q}{p+q} \left[ 1 - \left(1-\alpha \right)^{(1-b)T}\right] + (1-\alpha)^{(1-b)T} \left[ \frac{\qt}{\qt+\pt} \left( 1 - (1-\tilde{\alpha})^{bT} \right)+ \frac{r(0)}{N} (1-\tilde{\alpha})^{bT} \right]$.
%$M = \frac{\q}{\p+\q}\bbc{1-\bbc{1 - \alpha_t}^{(1-b)T}} + \bbc{\frac{\qt}{\pt+\qt}\bbc{1-\bbc{1 - \tilde{\alpha_t}}^{bT}} + \frac{r(0)}{N}\bbc{1 - \tilde{\alpha_t}}^{bT}}\bbc{1 - \alpha_t}^{(1-b)T}$.
%\end{proposition}
Recall that for $p=q$ (see \eqref{L_equalpq}), for the strategy $\st_L$, for $T$ sufficiently large, we get:
%\begin{small}
%\begin{eqnarray}
%\delta_N(T) &=& \frac{\tp}{\tp+\tq} + \left(\frac{1}{2} - \frac{\tp}{\tp+\tq} \right) \exp \{ -bT (\tp + \tq) \} \nonumber \\
%&+& \left( \delta_N(0) - \frac{1}{2} \right) \exp \{  -T[2\lambda p(1-b) + (\tp+\tq)b ]\} \nonumber \\
%&\approx & \frac{\tp}{\tp+\tq} + \left(\frac{1}{2} - \frac{\tp}{\tp+\tq} \right) (1-\tp -\tq)^{bT} \nonumber \\
%&+& \left( \delta_N(0) - \frac{1}{2} \right)(1-\tp-\tq)^{bT}(1-2\lambda p)^{(1-b)T} \nonumber \\
%\label{approxsoln}  
%\end{eqnarray} 
%\end{small}

\begin{small}
\begin{eqnarray}
\delta_N(T) &=& \frac{\widetilde{p}}{\widetilde{p}+\widetilde{q}} +\left(\frac{1}{2}-\frac{\widetilde{p}}{\widetilde{p}+\widetilde{q}}\right) e^{-bT \frac{\widetilde{p}+\widetilde{q}}{M}} \nonumber \\
    &+& \left(\delta_{N}(0)-\frac{1}{2}\right) e^{-\frac{T}{M}(2\lambda p(1-b)+b(\widetilde{p}+\widetilde{q}))} \nonumber \\
    &\approx & \frac{\tp}{\tp+\tq} + \left(\frac{1}{2} - \frac{\tp}{\tp+\tq} \right) \left(1-\frac{\tp+\tq}{M}\right)^{bT} \nonumber \\
&+& \left( \delta_N(0) - \frac{1}{2} \right)\left(1-\frac{\tp+\tq}{M}\right)^{bT}\left(1-\frac{2\lambda p}{M}\right)^{(1-b)T}. \nonumber \\
\label{approxsoln}  
\end{eqnarray} 
\end{small}
We denote this approximate solution by $\delta_N^{approx}$. Then we have the following result.

\begin{proposition}
	
For Model I defined in Section \ref{section:setting} and under Assumption \ref{rationalinfluence}, if the advertiser has a budget of $bT$ time-slots for $0 < b < 1$, given $\epsilon > 0$, $\exists \mu > 0$, which is a function of the model parameters $p, q, \tp, \tq, T, b$ and $\epsilon$ such that 
$$ P(|\delta_N(T) - \delta_N^{approx}| > \epsilon) < \mu. $$
\end{proposition}

\begin{proof} 
For $t \in [(1-b)T, T]$, we have
\begin{eqnarray*}
E[\delta_N(t+1)|\mathcal{F}_t] &=& \delta_N(t) + E[\chi(t+1)|\mathcal{F}_t]/M \\
&=& \delta_N(t) + [(1-\delta_N(t)) \tp - \delta_N(t) \tq]/M \\
&=& \tr \delta_N(t) + \tp/M,
\end{eqnarray*}
where $\tr = 1-(\tp+\tq)/M$.

Define $Y(t) = \tr^{-t} \delta_N(t) - \sum\limits_{k=1}^t \tp \tr^{-k}$. Then, for $t \leq T$, $Y(t)$ is a Martingale w.r.t the filtration $\{ \mathcal{F}_t \}_{t \geq 0}$, with bounded differences. In fact $|Y(t) - Y(t-1)| \leq \tr^{-T}(\tr+2)$. Then, by Azuma-Hoeffding, we have that for $\epsilon_1 > 0$,
\begin{equation} \label{mart1} P(|Y(T)-Y((1-b)T)| > \epsilon_1) < 2 \mu_1 \end{equation}
where, $\mu_1 = \exp \left( - \frac{\tr^T \epsilon_1^2}{2 bT (1+\tr)}\right)$.

Similarly, for $p=q$, $X(t) = r^{-t} \delta_N(t) - \sum\limits_{k=1}^t \lambda p r^{-k}$, where $r = 1-\lambda(p+q)/M = 1-2\lambda p/M$, is a Martingale with bounded differences for $t \leq T$. So, for $\epsilon_2 > 0$,
\begin{equation} \label{mart2} P(|X((1-b)T)-X(0)| > \epsilon_2) < 2\mu_2 \end{equation}
where, $\mu_2 = \exp \left( - \frac{r^T \epsilon_2^2}{2 (1-b)T (1+r)}\right)$.

For the sake of convenience, we write $Y(T) \stackrel{\epsilon_1}{\approx} Y((1-b)T)$ and $X((1-b)T) \stackrel{\epsilon_2}{\approx} X(0)$ with probability $2\mu_1$ and $2\mu_2$ respectively, instead of \eqref{mart1} and \eqref{mart2}. This implies, with probability $\mu_1$ and $\mu_2$, 
$$\delta_N(T) \stackrel{\epsilon'_1}{\approx} \tr^{bT} \delta_N((1-b)T) + \sum\limits_{k=(1-b)T+1}^T \tp \tr^{T-k}$$
$$\text{and }\delta_N((1-b)T) \stackrel{\epsilon'_2}{\approx} r^{(1-b)T} \delta_N(0) + \sum\limits_{k=1}^{(1-b)T} p \lambda r^{(1-b)T-k}. $$
respectively. Here $\epsilon'_1 = \tr^T \epsilon_1$ and $\epsilon'_2 = r^{(1-b)T} \epsilon_2$.

By the union bound, with probability $\mu_1+\mu_2$, for a suitable $\epsilon > 0$,
\begin{eqnarray}
\delta_N(T) & \stackrel{\epsilon}{\approx} & \tr^{bT}r^{(1-b)T} \delta_N(0) + \tr^{bT} \sum\limits_{k=1}^{(1-b)T} p \lambda r^{(1-b)T-k} \nonumber \\
&+& \sum\limits_{k=(1-b)T+1}^T \tp \tr^{T-k} \nonumber \\
& \stackrel{\epsilon}{\approx} & \tr^{bT}r^{(1-b)T} \delta_N(0) + \frac{\tp}{\tp+\tq} \left(1-\tr^{bT} \right) \nonumber \\
&+& \frac{\tr^{bT}}{2} (1-r^{(1-b)T}) \label{conc}.
\end{eqnarray}
It is easy to verify that the expressions in \eqref{approxsoln} and \eqref{conc} match. This concludes the proof.
\end{proof}

This generalizes the concentration result in \cite{MTNS} in one direction. However, note that in Proposition 2 of \cite{MTNS}, the assumption $p=q$ was not required. 
%From the above proof, we know that in the no external influence time period
%\begin{eqnarray*} 
%E[\delta_N(t+1) | \mathcal{F}_t] &=& \frac{\lambda p}{M} + \delta_N(t) \left(1- \frac{2\lambda p}{M}\right) \\
%&& +\frac{(p-q)}{M}(1-\delta_N(t)(1-\lambda)).
%\end{eqnarray*}
It is easy to check that for $p>q$, $X(t)$ defined above is a supermartingale, (while for $p<q$ it is a submartingale), however, since the corresponding solution of the O.D.E. is fairly complicated, it is difficult to conclude similar high-probability results for closeness of the solutions of the recursion and the corresponding O.D.E. We rely on the simulations (Figure~\ref{fig:tracking}) to demonstrate this. 

%%%%%%%%%%%%%%%%%%%%%%%%%%%%%%%%%%%%%%%%%%%%%%%%%%%%%%%%%%%%%%%%%%%%%%%%%%%%%%%%%%%%%%%%%%%%%%%%%%%%%%%%%%%%%%%%%%%%%%%%%%%%%%%%%%%%%%%%%%%%%%%%%%%%%%%%%%%%%%%%%%%%%%%%%%%%%%%%%%%%%%%%%%%%%%%%%%%%%%%%%%%%%%%%%%%%%%%%%%%%%%%%%%%%%%%%%%%%%%%%%%%%%%%%%%%%%%%%%%%%%%%%%%%%%%%%%%%%%%%%%%%%%%%%%%%%%%%%%%%%%%%%%%%%%%%%%%%%%%%%%%%%%%%%%%%%%%%%%%%%%%%%%%%%%%%%%%%%%%%%%%%%%%%%%%%%%%%%%%%%%%%%%%%%%%%%%%%%%%%%%%%%%%%%%%%%%%%%%%%%%%%%%%%%%%%%%%%%%%%%%%%%%%%%%%%%%%%%%%%%%%%%%%%%%%%%%%%%%%%%%%%%%%%%%%%%%%%%%%%%%%%%%%%%%%%%%%%%%%%%%%%%%%%%%%%%%%%%%%%%%%%%%%%%%%%%%%%%%%%%%%%%%%%%%%%%%%%%%%%%%%%%%%%%%%%%%%%%%%%%%%%%%%%%%%%%%%%%%%%%%%%%%%%%%%%%%%%%%%%%%%%%%%%%%%%%%%%%%%%%%%%%%%%%%%%%%%%%%%%%%%%%%%%%%%%%%%%%%%%%%%%%%%%%%%%%%%%%%%%%%%%%%%%%%%%%%%%%%%%%%%%%%%%%%%%%%%%%%%%%%%%%%%%%%%%%%%%%%%%%%%%%%%%%%%%%%%%%%%%%%%%%%%%%%%%%%%%%%%%%%%%%%%%%%%%%%%%%%%%%%%%%%%%%%%%%%%%%%%%%%%%%%%%%%%%%%%%%%%%%%%%%%%%%%%%%%%%%%%%%%%%%%%%%%%%%%%%%%%%%%%%%%%%%%%%%%%%%%%%%%%%%%%%%%%%%%%%%%%%%%%%%%%%%%%%%%%%%%%%%%%%%%%%%%%%%%%%%%%%%%%%%%%%%%%%%%%%%%%%%%%%%%%%%%%%%%%%%%%%%%%%%%%%%%%%%%%%%%%%%%%%%%%%%%%%%%%%%%%%%%%%%%%%%%%%%%%%%%%%%%%%%%%%%%%%%%%%%%%%%%%%%%%%%%%%%%%%%%%%%%%%%%%%%%%%%%%%%%%%%%%%%%%%%%%%%%%%%%%%%%%%%%%%%%%%%%%%%%%%%%%%%%%%%%%%%%%%%%%%%%%%%%%%%%%%%%%%%%%%%%%%%%%%%%%%%%%%%%%%%%%%%%%%%%%%%%%%%%%%%%%%%%%%%

\section*{Acknowledgements}

This work was supported in part by an Indo-French grant on ``Machine Learning for Network Analytics". The work of Neeraja Sahasrabudhe was also supported in part by the DST-INSPIRE Faculty Fellowship from the Govt. of India. The work of Sharayu Moharir was supported in part by a seed grant from IIT Bombay. 

\bibliographystyle{IEEEtran}
\bibliography{references}

\end{document}